\documentclass[lettersize,journal]{IEEEtran}
\usepackage{cite}
\usepackage{amsmath,amsfonts,amsthm}
\usepackage{algorithm}
\usepackage{algorithmic}

\usepackage{graphicx}
\usepackage{textcomp}
\usepackage{xcolor}
\usepackage{tabularx}
\usepackage{comment}
\usepackage{bm}
\usepackage[caption=false,font=footnotesize]{subfig}

\newtheorem{theorem}{Theorem}

\newcommand{\etal}{\textit{et~al}. }

\begin{document}

\title{Online Adaptive Computation Reuse in Collaborative Edge Computing:\\ A Two-Timescale Approach}

\author{Xingqiu~He, Chaoqun~You, Peng~Wang, Xiaoyu~Liu, and Tony~Q.~S.~Quek,~\IEEEmembership{Fellow,~IEEE}% 
  \thanks{X. He and C. You are with the Institute of Space Internet, Fudan University, Shanghai, China (emails: hexqiu@gmail.com, chaoqunyou@gmail.com).}
  \thanks{P. Wang is with the Information System Technology
and Design Pillar, Singapore University of Technology and Design, Singapore (e-mail: peng\_wang2@sutd.edu.sg).}
  \thanks{X. Liu is with the School of Communications and Information Engineering, Nanjing University of Posts and Telecommunications, Nanjing 210003, China (e-mail: xiaoyu.liu@njupt.edu.cn).}
  \thanks{T. Q. S. Quek is with the Singapore University of Technology and Design,
Singapore 487372, and also with the Yonsei Frontier Lab, Yonsei University,
South Korea (e-mail: tonyquek@sutd.edu.sg).}
}

% The paper headers
\markboth{Journal of \LaTeX\ Class Files,~Vol.~14, No.~8, August~2021}%
{Shell \MakeLowercase{\textit{et al.}}: A Sample Article Using IEEEtran.cls for IEEE Journals}

% \IEEEpubid{0000--0000/00\$00.00~\copyright~2021 IEEE}
% Remember, if you use this you must call \IEEEpubidadjcol in the second
% column for its text to clear the IEEEpubid mark.

\maketitle

% Computer Society journals place the abstract and keywords in the title area.
\begin{abstract}
Collaborative Edge Computing (CEC) is an efficient computing paradigm that enables neighboring edge servers to share computational resources with each other.
Although CEC can enhance resource utilization, it still suffers from duplicate computations because nearby end-users often offload tasks with similar inputs.
To improve system efficiency, the computation results of previously executed tasks can be cached and reused by subsequent tasks.
However, time-varying task popularity and arrival rates require caching and scheduling decisions to adapt to demand changes, while frequent cache updates incur additional costs.
To address this issue, this paper develops a two-timescale computation reuse algorithm for CEC networks.
We formulate an optimization problem that jointly considers weighted response time and cache update cost, with result caching decisions determined at the frame level and workload scheduling, cache searching, and computational resource allocation adjusted at the slot level.
Using recent workload observations, we construct a frame-level surrogate problem and decompose it into a caching subproblem and a scheduling subproblem.
For the caching subproblem, we introduce marginal storage efficiency and incorporate cache update costs into a bisection-based algorithm.
Under bounded normalized marginal sensitivity, the algorithm achieves a near-optimal objective value for the single-BS caching subproblem when individual result sizes are small relative to the cache capacity and the relaxed solution is sufficiently accurate.
For the scheduling subproblem, we utilize projected gradient descent and backtracking with warm starts across consecutive slots.
Numerical results demonstrate consistent performance gains over benchmark schemes
across diverse network and workload settings.
\end{abstract}


\section{Introduction}
In the past decade, the rapid evolution of machine learning and the Internet of Things has driven the
development of many computation-intensive applications, such as AR/VR~\cite{qiao2019web} and mobile gaming~\cite{lin2016cloudfog}.
To fulfill the stringent Quality of Service (QoS) requirements of these applications,
edge computing is regarded as a promising computing paradigm that provides computation-offloading services in proximity to end-users.
Nevertheless, due to limited resources at edge servers,
an edge computing system may suffer from severe performance degradation when dealing with bursty workloads~\cite{he2021providing}.
To address this issue, Collaborative Edge Computing (CEC)~\cite{tran2017collaborative}
has been proposed to alleviate load imbalance by offloading computation tasks from busy edge servers to idle ones.
Extensive studies~\cite{xiao2017qoe, chen2018computation, li2019learning, lyu2018distributed} have shown that CEC effectively improves the utilization of computational resources and thus
enhances system performance.

Despite the significant improvement CEC brings to resource utilization, computational resources can still be wasted due to duplicate computations.
As shown in~\cite{al2022promise}, clustered end-users are likely to offload similar computation tasks to edge servers due to temporal, spatial, or semantic correlations among input data.
For example, end-users in AR/VR may request the rendering of the same frame if they have similar poses~\cite{meng2020coterie},
and intelligent surveillance systems often need to analyze similar images captured within a short period of time~\cite{al2022promise}.
Based on this observation, \textit{computation reuse}~\cite{guo2018potluck, guo2018foggycache} has been proposed to further reduce resource waste.
The main idea is that the computation results of previously processed tasks can be cached and reused if subsequent tasks have similar inputs.
This mechanism not only avoids duplicate computations, but also reduces response time, as tasks with cached results can be completed immediately without being processed from scratch.

Computation reuse in CEC networks enables tasks at one edge server to leverage results cached at other servers, 
thereby enlarging the set of reusable results and reducing redundant caching~\cite{mastorakis2020icedge, al2022reservoir, al2024deduplicator}. 
The effectiveness of such collaboration critically depends on how well caching decisions and workload scheduling match the heterogeneous demands across edge servers. 
In practical CEC networks, however, task popularity and arrival rates are inherently time-varying, 
rendering caching and scheduling decisions optimized for the current demand profile less effective as the workload evolves. 
Adapting caching decisions to these changes incurs non-negligible communication and storage I/O overhead, making frequent reconfiguration undesirable.
In contrast, workload scheduling must react promptly to short-term traffic fluctuations.
This disparity in adaptation timescales motivates an online framework that jointly optimizes caching decisions and workload scheduling at different timescales.

Despite the aforementioned benefits, implementing computation reuse in CEC networks under dynamic demand is particularly challenging for three reasons.
First, compared to computation reuse at a single edge server, the mutual influence among caching decisions at different edge servers is hard to quantify.
Second, compared to conventional CEC networks, computation reuse changes the distribution of task processing time, thereby making caching and scheduling decisions deeply intertwined and difficult to optimize.
Third, caching decisions must balance the benefits of adapting to evolving demands against the associated reconfiguration costs, while making decisions online without knowledge of future workloads.

In this paper, we propose a two-timescale computation reuse algorithm for CEC networks.
Our objective is to simultaneously minimize the long-term weighted response time and cache update cost by jointly optimizing result caching, cache searching, workload scheduling, and computational power allocation.
The formulated problem is a mixed integer nonlinear programming problem and is proved to be NP-hard even for a time-invariant, single-server, single-application special case.
In order to address workload uncertainty, we construct a frame-level surrogate problem using recently observed workloads and apply alternating minimization to decompose it into a caching subproblem and a scheduling subproblem.
We incorporate the marginal update cost into the storage efficiency used to optimize each relaxed caching subproblem.
The resulting decisions are converted to binary candidates, which are accepted only when they reduce the overall surrogate objective.
Within each frame, the caching decisions remain fixed, while the scheduling decisions are dynamically adjusted to the task arrival rates.

% This paper extends our INFOCOM conference study~\cite{he2024exploiting}, which considered joint result caching and scheduling under time-invariant workloads.
% The journal version introduces time-varying demand and cache update costs, extends the caching and scheduling algorithms to operate on two timescales, and evaluates their adaptation to workload changes.
% It also makes the conditions for the caching analysis explicit and establishes a rounding-loss bound for the single-BS caching subproblem.
The main contributions are summarized as follows.
\begin{itemize}
    \item We formulate a two-timescale computation reuse problem that jointly accounts for response time and cache update cost in CEC networks.
        The model captures the coupling between frame-level result caching decisions and slot-level scheduling through cache-search overhead and queueing delay.
        We construct a frame-level surrogate problem from recent workload observations and decompose it into caching and scheduling subproblems.
    \item For the caching subproblem, we introduce marginal storage efficiency and combine its average over recent workload samples with the marginal update cost to develop a bisection-based algorithm.
        Our theoretical analysis shows that, under bounded normalized marginal sensitivity, the algorithm achieves a near-optimal objective value for the single-BS caching subproblem when individual result sizes are small relative to the cache capacity and the relaxed solution is sufficiently accurate.
        An acceptance rule ensures that binary reconstruction and numerical approximation do not increase the frame-level objective.
    \item For the scheduling subproblem, we develop a projected gradient descent algorithm, in which variable substitution is employed to transform the feasible region into probability simplices, thereby enabling efficient projection. A backtracking line search is incorporated to accelerate convergence and ensure solution stability, while warm-start initialization leverages the solution from the preceding slot to reduce computational overhead. We further establish the convergence of the frame-level objective sequence and characterize the computational complexity of the proposed algorithms at both timescales.
    \item We conduct extensive simulations under time-varying workloads to evaluate the response time, cache update cost, convergence behavior, and runtime overhead.
    The evaluation compares our algorithm against its static counterpart and representative caching and scheduling baselines, with a look-ahead scheme included as a performance reference. 
    The results demonstrate consistent performance gains over existing schemes across a wide range of network and workload settings.
\end{itemize}

The rest of the paper is organized as follows.
In Section \ref{section:related_work}, we review related work.
Section \ref{section:system_model} introduces the system model and formulates the corresponding optimization problem.
Section \ref{section:algorithm_design} describes the two-timescale algorithm and its theoretical analysis.
The numerical results and conclusions are presented in Section \ref{section:simulation} and Section \ref{section:conclusion}, respectively.

\section{Related Work} \label{section:related_work}
\subsection{Scheduling and Caching in CEC Networks}
Workload scheduling and resource allocation in CEC networks have been extensively investigated to balance computational loads across edge servers and improve overall system performance. Early studies considered distributed optimization to jointly account for user experience and edge-server power efficiency~\cite{xiao2017qoe}, while online peer offloading was developed to minimize long-term latency subject to time-average energy constraints~\cite{chen2018computation}. Distributed optimization has also been employed to derive cost-effective computation offloading decisions~\cite{lyu2018distributed}. Beyond these objectives, deadline-aware scheduling schemes have been investigated to accommodate the stringent latency requirements of computation tasks~\cite{li2019learning, he2021providing}.

Complementary to workload scheduling, collaborative content caching improves resource utilization by exploiting storage resources distributed across edge servers. 
Existing studies have optimized content caching decisions by accounting for content popularity~\cite{ostovari2016efficient, wang2017multi}, spectrum efficiency~\cite{zhang2017cooperative}, as well as network topology and user demands~\cite{saputra2019novel}. More recent studies have jointly considered content caching, request routing, and load balancing in heterogeneous edge networks~\cite{peng2024cooperative}, while online caching schemes based on contextual bandits and federated learning have been developed to accommodate emerging content and facilitate cache sharing among edge servers~\cite{chen2024caching}. Unlike content caching, which primarily reduces the cost of retrieving existing content, computation reuse eliminates redundant task execution by caching and reusing previously computed results, thereby coupling caching decisions with computation scheduling.

\subsection{Computation Reuse}
Similarities among task inputs create opportunities to avoid redundant computation by identifying and sharing reusable results. Representative systems include Potluck~\cite{guo2018potluck} and FoggyCache~\cite{guo2018foggycache}, which exploit computation reuse across applications and devices, respectively. Subsequent studies have investigated cache management and reuse-aware task execution at the edge~\cite{lee2019case, bellal2021coxnet, al2022promise}. ICedge~\cite{mastorakis2020icedge} and Reservoir~\cite{al2022reservoir} leverage information-centric networking to enable computation reuse across distributed edge servers, while Deduplicator~\cite{al2024deduplicator} integrates similarity-aware request routing with load balancing to coordinate reuse across multiple servers. Collectively, these systems establish practical mechanisms for discovering, retrieving, and sharing reusable computation results and for directing requests to appropriate edge servers.

Complementing these system-oriented efforts, several studies have investigated the optimization of computation reuse in edge networks. Whispering~\cite{nour2021whispering} considers service offloading and computation reuse in cloud--edge networks, while LOAM~\cite{zhang2025loam} jointly optimizes request forwarding, data and result caching, and computation placement by accounting for congestion-dependent costs and elastic cache sizes.
Zhu \etal~\cite{zhu2026optimized} optimize data-block offloading, execution ordering, and result caching to reduce delay and caching costs in cooperative edge networks.
Our preliminary work~\cite{he2024exploiting} jointly optimized computation result caching and workload scheduling under time-invariant workloads.
To accommodate the workload evolution in practical situations, 
this paper develops an online two-timescale framework that coordinates frame-level result caching with slot-level cache searching, workload scheduling, and computational resource allocation.

\section{System Model and Problem Formulation} \label{section:system_model}

In this section, we describe the considered CEC system with computation reuse and formulate a two-timescale response time minimization problem.
Complexity analysis shows that the formulated problem is NP-hard even for a significantly simplified time-invariant special case.

\subsection{Computation Reuse Model}
\begin{figure}[t]
\centering
\includegraphics[width=0.4\textwidth]{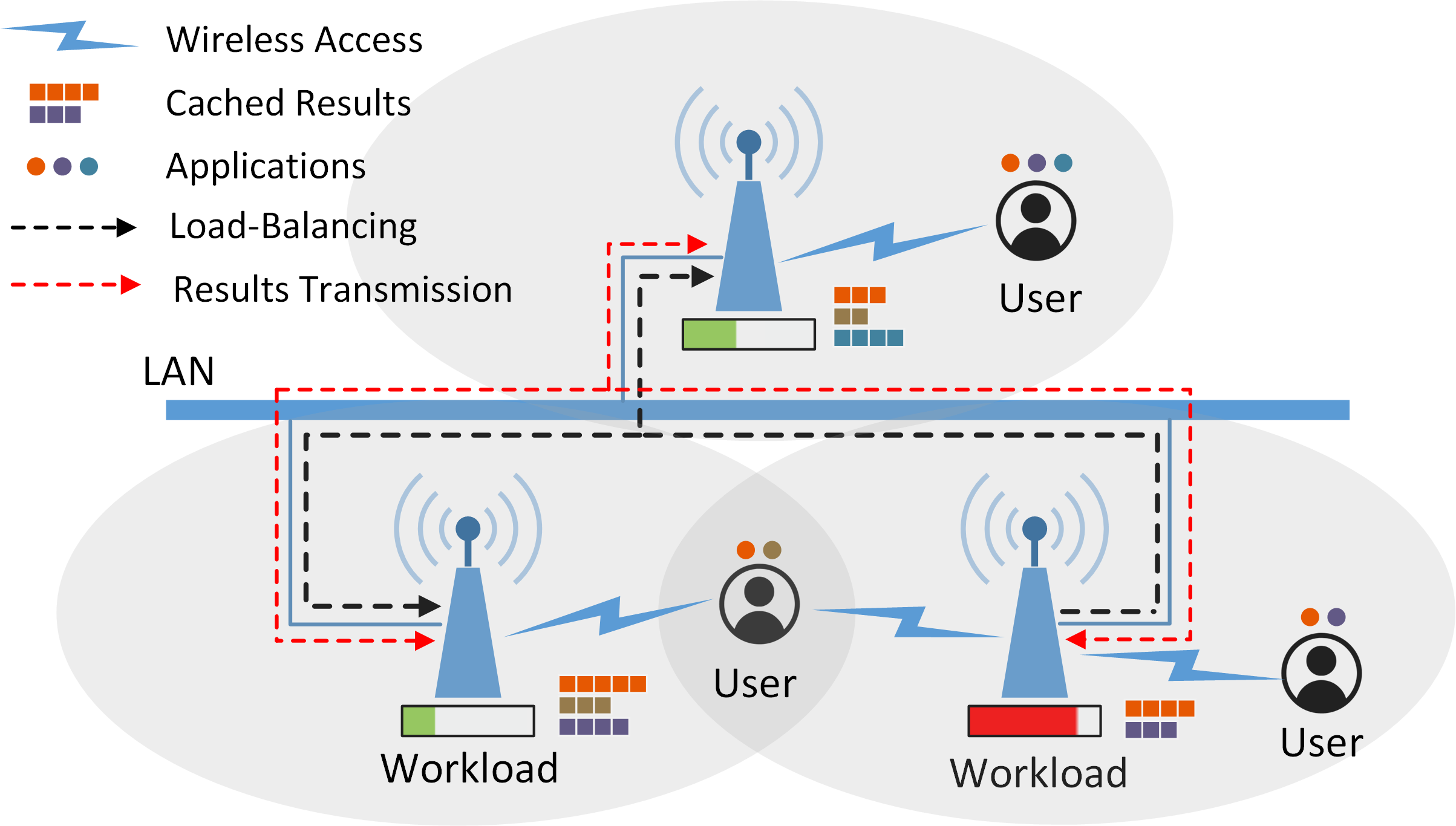}
\caption{A simple example of the considered system model.}
\label{fig:system}
\end{figure}
As shown in Fig.\ref{fig:system}, we consider a CEC network with $N$ BSs indexed by $\mathcal{N} = \{1, 2, \dots, N\}$. 
Each BS $n\in\mathcal{N}$ is co-located with an edge server with computational capacity $C_n^c$ (in CPU frequency) and storage capacity
% \footnote{In current practice, the computation results are usually stored in RAM so $C^s_n$ refers to the RAM space. But our system model and proposed algorithm also apply to disk storage situations.} 
$C_n^s$ (in GB), respectively.
Suppose there are $A$ different applications indexed by $\mathcal{A} = \{1, 2,\dots, A \}$ running in the CEC network. 
The end-users running these applications may offload computation tasks to nearby BSs with different input data.
%comes from multiple users in the network, and each application generates a sequence of tasks and then offloads them to one or more BSs in the same area to be processed. 
Taking an AR/VR application for example, the input data of each task indicate the user's pose, comprised of 3D position and 3D orientation.
The output of the task is the rendered frame to be displayed on the user's head-mounted displays.
%\textcolor{blue}{When a new task is offloaded to BS $n$, there are two choices: (i) the task uses the cached results on the edge server as its outputs (i.e., computation reuse), and (ii) the task is processed directly by the edge server without searching for the cached computation results.}

The system operates on two timescales. Specifically, the time horizon is divided into frames indexed by $r\in\mathcal{R}=\{1,2,\dots\}$,
and each frame consists of $T$ slots denoted by $\mathcal{T}_r=\{(r-1)T+1,(r-1)T+2,\dots,rT\}$.
At the beginning of each frame, the task popularity is estimated from recent observations, and the computation result caching decisions are made accordingly.
These caching decisions remain unchanged within the frame,
since updating the cached results usually incurs non-negligible communication and storage costs.
In contrast, the workload scheduling policy can be adjusted in every slot
according to the current task arrival rates.
We assume the task popularity remains approximately unchanged within a frame, while the task arrival rates remain approximately unchanged within a slot.
% This assumption captures the fact that the task popularity usually evolves much more slowly than the instantaneous workload.
For notational simplicity, we hereafter omit the corresponding index sets from universal quantifiers. Specifically, $\forall a,m,n,k,r,t$ is understood as $\forall a\in\mathcal{A},\,m,n\in\mathcal{N},\,k\in\mathcal{K}^a,\,r\in\mathcal{R},\,t\in\mathcal{T}_r$, with analogous interpretations applying when only a subset of these indices appears.

According to the results in \cite{al2022reservoir}, 
for a given application, computation tasks with sufficiently similar input data usually have the same outputs,
so the results of previous tasks can be reused by subsequent tasks. 
Specifically, for application $a\in\mathcal{A}$, suppose we have identified a set of typical inputs $\mathcal{K}^a = \{1,2,\dots, K^a\}$ from its historical data. 
If a new task of application $a$ is offloaded from end-users, we can first find the closest typical input using fast and space-efficient mechanisms, such as Locality Sensitive Hashing (LSH)~\cite{datar2004locality} and Feature Hashing (FH)~\cite{weinberger2009feature}. 
After that, we compute the similarity between the two inputs according to the chosen metric (e.g., structural similarity \cite{wang2004image} or cosine similarity \cite{singhal2001modern}). 
If the similarity exceeds a certain threshold and the output of the corresponding typical input has been cached, the computation result can be directly returned without processing the task from scratch.

Although computation reuse saves a significant amount of computational resources when the target result is cached, 
it also incurs extra computational costs due to operations like hashing and similarity calculation. 
Therefore, in some cases, it may be beneficial to directly process tasks without searching for the cached results, 
especially when the probability of successfully finding the target result is low. 
To allow for a flexible cache-searching scheme, we introduce a binary variable $y_n^a[t] \in \{0, 1\}$ to indicate whether 
we will search for a cached result for computation tasks from application $a$ at BS $n$ in slot $t$. 
If $y^a_n[t]=0$, then all tasks of application $a$ will be computed from scratch at BS $n$ in slot $t$, while $y^a_n[t]=1$ indicates otherwise.

Clearly, the optimal value of $y^a_n[t]$ depends on how often we successfully find the target result.
Suppose we are processing a task at BS $n$ in frame $r$, then we can define the \textit{local hit rate} $P_{lhr}^{a,n}[r]$ and \textit{neighboring hit rate} $P^{a,n}_{nhr}[r]$ as the probabilities that the desired computation result is cached at BS $n$ and the neighboring BSs, respectively.
For any typical input $k\in\mathcal{K}^a$, let $x^a_{nk}[r]\in\{0,1\}$ indicate whether the output of $k$ is cached at BS $n$ in frame $r$,
and let $p_k^a[r]$ denote the probability that a task from application $a$ is sufficiently similar to $k$ in frame $r$.
Then the local and neighboring hit rate can be expressed as $P_{lhr}^{a,n}[r] = \sum_{k\in\mathcal{K}^a} x_{nk}^a[r] p_k^a[r]$
and $P^{a,n}_{nhr}[r] = \sum_{k\in\mathcal{K}^a} p^a_k[r](1-x^a_{nk}[r])I_{>0}(\sum_{m\in\mathcal{N}, m\ne n}x^a_{mk}[r])$,
where $I_{>0}(\cdot)$ is the indicator function that returns $1$ if the input value is positive.
For convenience, we define $P^{a,n}_{hr}[r] = P^{a,n}_{lhr}[r] + P^{a,n}_{nhr}[r]$ as the aggregate hit rate.
Apparently, $P^{a,n}_{hr}[r]$ is the hit rate of application $a$ in the whole system and remains the same for different BS $n$.
Therefore, we can drop the superscript $n$ and use $P^{a}_{hr}[r]$ afterward.

Although it is desirable to enhance the hit rate as much as possible, the number of cached results is constrained by the storage capacities of BSs.
Suppose the data size of the computation result corresponding to typical input $k$ is $s_k^a$, then the cached data at BS $n$ should not exceed its storage capacity, which can be expressed as
\begin{equation}
    \sum_{a\in\mathcal{A}} \sum_{k\in\mathcal{K}^a} x^a_{nk}[r]s^a_k \leq C^s_n, \quad \forall n,r. \label{cons:storage}
\end{equation}

Updating the cached results between two consecutive frames consumes network bandwidth and storage I/O resources.
Let $c_n^u$ be the normalized update cost per unit of data at BS $n$, and let $[z]^+=\max\{z,0\}$.
Then the cache update cost in frame $r$ can be expressed as
\begin{equation}
    U[r] = \sum_{n\in\mathcal{N}}\sum_{a\in\mathcal{A}}\sum_{k\in\mathcal{K}^a}
    c_n^u s_k^a[x_{nk}^a[r]-x_{nk}^a[r-1]]^+, \label{eq:update_cost}
\end{equation}
where $\bm{x}[0]$ represents the initial caching decisions before the first frame.
Notice that \eqref{eq:update_cost} only counts newly cached results, since removing a cached result incurs negligible cost.

\subsection{Task Scheduling Model}
In this subsection, we introduce the decision variables related to task scheduling and derive the expression of response time.
Let $D^a[r,t]$ be the average response time of application $a$ in slot $t\in\mathcal{T}_r$. Generally, $D^a[r,t]$ consists of two parts: the computation delay that is caused by processing the computation tasks and the transmission delay that is caused by transmitting tasks and computation results among BSs. In the following, we will provide the formulations of computation delay and transmission delay, respectively.

Assume that the task arrival at BS $n$ in slot $t$ follows a Poisson process with expected rate $R_n^a[t]$.
We adopt a quasi-stationary queueing model in each slot, which is reasonable when the slot duration is sufficiently long compared with the task processing time.
For convenience, we define the total arrival rate of all BSs in the system as $R^a[t] = \sum_{n\in\mathcal{N}} R^a_n[t]$. 
%Given that each task has two choices when it arrives at BS $n$, 
In CEC, the computation workload can be redistributed among collaborative BSs.
Therefore, we introduce $\lambda_n^a[t] \in [0,1]$ to denote the fraction of application $a$'s workload distributed to BS $n$ in slot $t$. Apparently, for any application $a$, the sum of its workload distribution across all BSs should be equal to $1$, hence we have
\begin{equation}
    \sum_{n\in\mathcal{N}} \lambda^a_n[t] = 1, \quad \forall a,r,t. \label{cons:scheduling}
\end{equation}

As discussed in the previous subsection, the processing procedures of tasks are different depending on the value of $y^a_n[t]$.
Therefore, we will derive the expression of computation delay case by case.
Let $f^a_n[t]$ be the computational capability (in CPU frequency) allocated to application $a$ at BS $n$ in slot $t$
and $w^a$ be the average workload (in CPU cycles) of tasks from application $a$. 
As in~\cite{xu2018joint, ma2020cooperative}, we assume the workload of tasks follows an exponential distribution.
When $y_n^a[t] = 0$, the computation process satisfies the M/M/1 queueing model, so the average computation delay without cache-searching is
\begin{equation}
    D^{a,0}_n[t] = \frac{1}{\mu^{a,0}_n[t] - \lambda^a_n[t] R^a[t]}, \label{eq:delay_no_cache}
\end{equation}
where $\mu^{a,0}_n[t] = f^a_n[t]/w^a$ is the average processing rate with respect to the number of tasks.
To guarantee the stability of the queue, we must ensure
\begin{equation}
    \lambda^a_n[t] R^a[t] < \mu^{a,0}_n[t], \quad \forall a,n,r,t. \label{cons:stability_0}
\end{equation}

When $y^a_n[t]=1$, we will first search the cached results before processing. 
If the task's computation result is cached and reused, then the computation cost of the task is incurred by the searching process, which is denoted by $w^s$ (in CPU cycles). 
Otherwise, the task still needs to be processed from scratch.
Consequently, the Cumulative Distribution Function (CDF) of the total computation cost $w$ can be derived as
\begin{equation*}
    F(w) = 
    \begin{cases}
        1 - (1-P^{a}_{hr}[r])e^{-(w-w^s)/w^a}, & w \geq w^s \\
        0, & w < w^s
    \end{cases},
\end{equation*}
where $F(w)$ jumps from $0$ to $P^{a}_{hr}[r]$ at $w^s$, and gradually increases to $1$ thereafter.

According to the M/G/1 queueing model, the average computation delay with cache-searching is
\begin{align}
    D^{a,1}_n[r,t] = &\frac{1}{\mu^{a,1}_n[r,t]} + \frac{\lambda^a_n[t]R^a[t]}{2\mu^{a,1}_n[r,t](\mu^{a,1}_n[r,t]-\lambda^a_n[t]R^a[t])} \notag \\
    &+ \frac{(1-P^{a}_{hr}[r])(1+P^{a}_{hr}[r])\lambda^a_n[t]R^a[t]\mu^{a,1}_n[r,t]}{2(\mu^{a,1}_n[r,t]-\lambda^a_n[t]R^a[t])(\mu^{a,0}_n[t])^2}
    \label{delay:cache}
\end{align}
where $\mu^{a,1}_n[r,t] = f^a_n[t]/(w^s+(1-P^{a}_{hr}[r])w^a)$.
To guarantee the stability of the queue, we must have
\begin{equation}
    \lambda^a_n[t] R^a[t] < \mu^{a,1}_n[r,t], \quad \forall a,n,r,t.
    \label{cons:stability_1}
\end{equation}
Combining \eqref{cons:stability_0} and \eqref{cons:stability_1} yields
\begin{equation}
    \begin{aligned}
    &\lambda^a_n[t] R^a[t] - (1-y^a_n[t])\mu^{a,0}_n[t] \\
    &\qquad - y^a_n[t]\mu^{a,1}_n[r,t] < 0, \quad \forall a,n,r,t.
    \end{aligned}
    \label{cons:stability}
\end{equation}

Notice that if the target result is cached in other BSs, it must be transmitted to the BS $n$ for reuse. As in \cite{ma2020cooperative}, the transmission delay from neighboring BSs to BS $n$ can be assumed to be a constant $D^t_n$.
Hence, the average processing delay of application $a$ at BS $n$ in slot $t\in\mathcal{T}_r$ is given by
\begin{equation}
    D^a_n[r,t] = (1-y^a_n[t])D^{a,0}_n[t] + y^a_n[t](D^{a,1}_n[r,t] + P^{a,n}_{nhr}[r]D^t_n).
    \label{eq:D^a_n}
\end{equation}

In addition to the processing delay $D_n^a[r,t]$, the response time $D^a[r,t]$ of application $a$ also includes the transmission delay caused by transmitting tasks and computation results among BSs.
If $\lambda^a_n[t]R^a[t] - R^a_n[t] > 0$, then there are extra tasks offloaded from other BSs to BS $n$ and the induced transmission delay is
$D^t_n(\lambda^a_n[t]R^a[t] - R^a_n[t])/R^a[t]$.
Otherwise, if $\lambda^a_n[t]R^a[t] - R^a_n[t] < 0$, then part of BS $n$'s tasks will be offloaded to the other BSs.
The transmission delay induced by returning computation results to BS $n$ is $D^t_n(R^a_n[t] - \lambda^a_n[t] R^a[t])/R^a[t]$.
Based on the above discussion, the expression of $D^a[r,t]$ is given as
\begin{equation*}
    D^a[r,t] = \sum_{n\in\mathcal{N}} \left[ \lambda^a_n[t]D^a_n[r,t] + \frac{|\lambda^a_n[t] R^a[t] - R^a_n[t]|}{R^a[t]}D^t_n \right].
\end{equation*}

In practice, it may be beneficial to offload part of tasks to the cloud when the local BSs are overloaded.
This situation can be accommodated in our model by regarding the remote cloud as a neighboring BS with sufficient computational power and large transmission delays.
If the bandwidth of the core network is limited, one can also treat the transmission process as an M/M/1 queueing model (as in \cite{ma2020cooperative}).
The resulting transmission delay is a function of the offloaded workload, which takes a similar form of \eqref{eq:delay_no_cache},
and hence brings no extra complexity to the problem-solving.
Therefore, to keep the description concise and focus on the key parts of the problem, we do not explicitly consider offloading to the cloud in our model.

\subsection{Problem Formulation}
This paper aims to minimize a long-term cost comprising the weighted response time across different applications and the cache update cost by jointly optimizing the frame-level computation result caching decisions $x^a_{nk}[r]$ and the slot-level workload scheduling policy $\lambda^a_n[t]$, 
cache-searching strategy $y^a_n[t]$, and computational capability allocation $f^a_n[t]$.
Let $\phi^a$ be the weight of application $a$ that reflects its priority,
and the coefficient $\eta\geq 0$ controls the tradeoff between response time and cache update cost.
Then the corresponding two-timescale optimization problem can be formulated as:
\begin{align}
    \min_{\bm{x},\bm{y},\bm{\lambda},\bm{f}}\; & \limsup_{J\rightarrow\infty}\frac{1}{J}\sum_{r=1}^{J}
    \mathbb{E}\left[\frac{1}{T}\sum_{t\in\mathcal{T}_r}\sum_{a\in\mathcal{A}} \phi^a D^a[r,t] + \eta U[r]\right] \label{lmp} \\
    s.t.\quad & \eqref{cons:storage}, \eqref{cons:scheduling}, \eqref{cons:stability} \quad \tag{\ref{lmp}{a}} \label{lmp:before_cons} \\
    & \sum_{a\in\mathcal{A}} f^a_n[t] \leq C^c_n, \quad \forall n,r,t \tag{\ref{lmp}{b}} \label{lmp:computing} \\
    & x^a_{nk}[r]\in\{0,1\}, \quad \forall a,n,k,r \tag{\ref{lmp}{c}} \label{lmp:variable_x} \\
    % & y^a_n[t]\in\{0,1\}, \quad \forall a,n,r,t \tag{\ref{lmp}{d}} \label{lmp:variable_y} \\
    & y^a_n[t]\in\{0,1\}, \lambda^a_n[t], f^a_n[t] \geq 0, \quad \forall a,n,t \tag{\ref{lmp}{d}} \label{lmp:variable_ylf}
\end{align}
where $\bm{x}, \bm{y}, \bm{\lambda}, \bm{f}$ are vectors of corresponding decision variables.
The expectation is taken with respect to the time-varying task popularity and arrival rates.
% At the beginning of frame $r$, $\bm{x}[r]$ is determined according to the current popularity estimate and the cache state inherited from the previous frame.
% At the beginning of slot $t$, $\bm{y}[t]$, $\bm{\lambda}[t]$, and $\bm{f}[t]$ are determined after the current arrival rates are observed.
Constraint \eqref{lmp:before_cons} contains the caching and scheduling constraints described above.
Constraint \eqref{lmp:computing} requires that the computational power allocated to applications cannot exceed the BS's computational capacity.
Constraint \eqref{lmp:variable_x} and \eqref{lmp:variable_ylf} define the feasible regions of control variables.
% Notice that when the task popularity and arrival rates are time-invariant, $\eta=0$, and all the decisions remain unchanged over time,
% problem \eqref{lmp} reduces to the conventional static joint caching and scheduling problem.

The formulated problem \eqref{lmp} is a mixed integer non-linear programming problem and has no polynomial-time solution in general.
In fact, we can show that the problem \eqref{lmp} is NP-hard even for a significantly simplified special case, as stated in the following theorem:
\begin{theorem}
    The problem \eqref{lmp} is NP-hard even if all system states are time-invariant and there is only one BS and one application in the system.
    \label{theorem:np}
\end{theorem}
\begin{proof}
   Please see Appendix \ref{appendix:1}.
\end{proof}

Theorem \ref{theorem:np} demonstrates that it is computationally prohibitive to compute the optimal caching decisions even for the time-invariant single-BS and single-application case.
In realistic situations, we also need to consider the correlation between caching decisions at different BSs,
the workload scheduling policy among BSs, the allocation of computational power to each application, and the temporal variations of task popularity and arrival rates.
All of these aspects substantially complicate the formulated problem \eqref{lmp} and make it challenging to solve.

\section{Algorithm Design and Theoretical Analysis} \label{section:algorithm_design}
In this section, we develop an online algorithm for problem \eqref{lmp} by coordinating caching and scheduling decisions across two timescales.
At the frame timescale, caching decisions are determined using current task popularity estimates and recent workload observations.
At the slot timescale, workload scheduling, cache searching, and computational resource allocation are adapted to the current task arrival rates, with caching decisions held fixed within each frame.
We further analyze the convergence properties and computational complexity of the proposed algorithm.

\subsection{Problem Decomposition}
At the beginning of frame $r$, result caching decisions must be determined without knowledge of future task arrival rates.
In practice, task popularity and arrival rates often exhibit temporal correlation, making recent observations informative about near-future demand.
We therefore characterize the workload and task popularity using observations from the $L$ most recent slots, indexed by
$\mathcal{L}_r=\{(r-1)T-L+1,\dots,(r-1)T\}$.
Each sample $\ell\in\mathcal{L}_r$ is represented by the arrival-rate vector $\bm{R}[\ell]=(R^a_n[\ell])_{a\in\mathcal{A},n\in\mathcal{N}}$,
with aggregate application arrival rate $R^a[\ell]=\sum_{n\in\mathcal{N}}R^a_n[\ell]$.
Given the caching decisions $\bm{x}[r]$ and the estimated task popularity $\bm{p}[r]$, let $D^{a,\ell}[r]$ denote the response time of application $a$ evaluated under workload sample $\bm{R}[\ell]$ and sample-specific scheduling decisions
$\bm{y}^{\ell}[r],\bm{\lambda}^{\ell}[r],\bm{f}^{\ell}[r]$.
The resulting sample-average weighted response time for frame $r$ is expressed as
\begin{equation}
    \overline{D}[r] = \frac{1}{L}\sum_{\ell\in\mathcal{L}_r}\sum_{a\in\mathcal{A}}\phi^aD^{a,\ell}[r].
    \label{eq:sample_average_delay}
\end{equation}
Based on the observed samples, we formulate the following frame-level surrogate problem:
\begin{align}
    \min_{\bm{x}[r],\{\bm{y}^{\ell}[r],\bm{\lambda}^{\ell}[r],\bm{f}^{\ell}[r]\}_{\ell\in\mathcal{L}_r}}\quad
    & \overline{D}[r]+\eta U[r] \label{subproblem:frame} \\
    s.t.\quad & \eqref{cons:storage}, \eqref{cons:scheduling}, \eqref{cons:stability},
    \eqref{lmp:computing}-\eqref{lmp:variable_ylf}. \notag
\end{align}
In problem \eqref{subproblem:frame}, the slot-level constraints are applied to every workload sample by replacing the slot-indexed arrival rates and scheduling decisions with their sample-specific counterparts.

For convenience, we temporarily relax constraint \eqref{cons:stability} and will re-consider it in Section \ref{subsection:scheduling}.
To optimize the weighted response time, we should choose the cache-searching strategy that minimizes the processing delay
for each BS and each application.
According to \eqref{eq:D^a_n}, when both processing modes satisfy the corresponding stability conditions, the optimal value of $y^a_n[t]$ is $y^a_n[t]=1$ if
\begin{equation}
    D^{a,0}_n[t] > D^{a,1}_n[r,t] + P^{a,n}_{nhr}[r]D^t_n,
    \label{eq:optimal_y}
\end{equation}
and $y^a_n[t]=0$ otherwise.
If only one processing mode is stable, we select the stable one.
Therefore, $\bm{y}[t]$ is automatically determined once $\bm{x}[r], \bm{\lambda}[t]$, and $\bm{f}[t]$ are resolved.
Hence, we will not consider $\bm{y}$ explicitly in the following optimization.

The nonconvex coupling between caching and scheduling decisions in problem \eqref{subproblem:frame} makes joint optimization challenging.
We therefore adopt alternating minimization (AM)\footnote{Also known as block coordinate descent.} to update the frame-level caching decisions and the sample-specific scheduling decisions in turn.
In each caching update, we fix the sample-specific scheduling decisions, including the cache-searching decisions, and minimize the sum of the sample-average weighted response time and cache update cost.
The relaxed solution is converted to a binary candidate, after which the cache-searching decisions are updated according to \eqref{eq:optimal_y}.
To account for rounding and numerical errors, we accept the candidate only if it strictly decreases the complete objective of problem \eqref{subproblem:frame}.
With the caching decisions fixed, the scheduling decisions are separable across workload samples and can be optimized independently.

Once the caching decisions $\bm{x}[r]$ are implemented, they remain fixed throughout frame $r$, while the scheduling decisions are adapted to the observed task arrival rates in each slot.
Since the cache update cost is fixed within the frame, the scheduling problem for slot $t\in\mathcal{T}_r$ reduces to
\begin{align}
    \min_{\bm{\lambda}[t],\bm{f}[t]}\quad & \sum_{a\in\mathcal{A}} \phi^a D^a[r,t] \label{subproblem:scheduling} \\
    s.t.\quad & \eqref{cons:scheduling}, \eqref{cons:stability}, \eqref{lmp:computing}, \eqref{lmp:variable_ylf}. \notag
\end{align}
where the cache-searching decisions $\bm{y}[t]$ are determined by \eqref{eq:optimal_y} for the given caching decisions $\bm{x}[r]$ and scheduling variables $\bm{\lambda}[t]$ and $\bm{f}[t]$.
In the following two subsections, we develop efficient algorithms for the caching and scheduling subproblems.

\subsection{Caching Subproblem} \label{subsection:caching}

The caching decisions across BSs are tightly coupled, making their joint optimization challenging. 
To address this issue, we optimize the caching decision of one BS at a time while fixing those of the remaining BSs. 
Specifically, let $\bm{x}_n$ and $\bm{x}_{-n}$ denote the caching decisions at BS $n$ and all other BSs, respectively. 
For notational simplicity, we omit the frame index $r$ throughout this subsection. 
For any given $\bm{x}_{-n}$, we optimize $\bm{x}_n$ while keeping all sample-specific scheduling and cache-searching decisions fixed. 
Let $x_{nk}^{a,\mathrm{old}}=x_{nk}^a[r-1]$ denote the caching decision in the preceding frame.
Since $x_{nk}^{a,\mathrm{old}}\in\{0,1\}$, the cache update cost can be expressed as a linear function of $x_{nk}^a$:
\begin{equation}
    [x_{nk}^a-x_{nk}^{a,\mathrm{old}}]^+
    =(1-x_{nk}^{a,\mathrm{old}})x_{nk}^a.
    \label{eq:update_linearization}
\end{equation}
Therefore, by defining $\delta_{nk}^a=\eta c_n^u(1-x_{nk}^{a,\mathrm{old}})$ and omitting the constant update costs at other BSs, we obtain the caching subproblem
\begin{align}
    \min_{\bm{x}_n}\quad &J_n(\bm{x}_n)
      =\overline D(\bm{x}_n,\bm{x}_{-n})
       +\sum_{a,k}\delta_{nk}^a s_k^a x_{nk}^a
       \label{subproblem:caching}\\
    \text{s.t.}\quad &\sum_{a,k}s_k^a x_{nk}^a\leq C_n^s,
      \quad x_{nk}^a\in\{0,1\}.\notag
\end{align}
For convenience, we temporarily relax the binary constraint of $x^a_{nk}$ and assume it is a real number in $[0,1]$.
The relaxed caching subproblem at BS $n$ is called problem $P_n^c$.
We will reconstruct a binary caching decision after solving $P_n^c$.
% With $\bm{x}_{-n}$ binary, the hit rates are extended as affine functions of $\bm{x}_n$ that coincide with the original expressions at binary decisions.
% We assume that the system remains stable when BS $n$ caches no results, i.e., $\bm{x}_n=\bm{0}$, under the fixed decisions specified above.
% Since increasing the hit rate preserves stability, all relaxed caching decisions and their rounded solutions then remain stable.
% If this condition is not satisfied, we retain the current caching decisions at BS $n$.

According to the expressions of $D^a_n$ and $D^{a,1}_n$, the influence of caching decisions on response time is through the hit rate.
However, the caching decisions at the same BS are constrained by the storage space, and caching a new result incurs an update cost.
To use the available storage efficiently, it should be allocated to the results that contribute most to reducing $J_n$ in \eqref{subproblem:caching}.
In the following, we first quantify the reduction in response time and then incorporate the update cost.
Inspired by this idea, we use $\sigma^{a}_{nk} = x^{a}_{nk} s^a_k$ to denote the storage space allocated to $x^a_{nk}$ and
define the marginal storage efficiency of $x^a_{nk}$ under workload sample $\ell$ as
\begin{align}
    \epsilon^{\ell}(x^a_{nk}) &= \frac{\partial \sum_{a'\in\mathcal{A}} \phi^{a'} D^{a',\ell}}{\partial \sigma^{a}_{nk}} \notag \\
    &= \frac{\partial \sum_{a'\in\mathcal{A}} \sum_{n'\in\mathcal{N}} \phi^{a'} \lambda^{a'}_{n'} y^{a'}_{n'} (D^{a',1}_{n'} + P^{a',n'}_{nhr} D^t_{n'})}
    {\partial \sigma^a_{nk}} \notag \\
    &= \sum_{n'\in\mathcal{N}} \phi^a \lambda^a_{n'} y^a_{n'} \left( \frac{\partial D^{a,1}_{n'}}{\partial P^{a}_{hr}} 
    \frac{\partial P^a_{hr}}{\partial x^a_{nk}s^a_k} + D^t_{n'} \frac{\partial P^{a,n'}_{nhr}}{\partial x^a_{nk}s^a_k}  \right). \label{eq:x_a_nk}
\end{align}
% A more negative $\epsilon^{\ell}(x^a_{nk})$ indicates a greater reduction in the weighted response time per unit of additional storage.
For notational convenience, we omit the sample index $\ell$ in the following derivation.
Let $I^{a}_{-n,k} = I_{>0}(\sum_{m\in\mathcal{N},m\ne n} x^a_{mk})$ be the binary indicator indicating whether the result of the 
$k$-th typical input in $\mathcal{K}^a$ is cached outside BS $n$.
Then according to the definition of $P^a_{hr}$ and $P^{a,n}_{nhr}$ we have 
\begin{align}
    \frac{\partial P^a_{hr}}{\partial x^a_{nk}} &= p^a_k(1-I^a_{-n,k}), \label{eq:P_a_hr_1} \\
    \frac{\partial P^{a,n'}_{nhr}}{\partial x^a_{nk}} &= 
    \begin{cases}
        p^a_k(1 - x^a_{n'k})(1 - I^a_{-n,k}) & n' \ne n, \label{eq:P_a_nhr_2} \\
        -p^a_k I^a_{-n',k} & n' = n.
    \end{cases}
\end{align}
Since $x^a_{n'k} = 1$ implies $I^a_{-n,k} = 1$ when $n' \ne n$, the term $p^a_k(1 - x^a_{n'k})(1 - I^a_{-n,k})$ can be simplified to
$p^a_k(1 - I^a_{-n,k})$. Substituting \eqref{eq:P_a_hr_1} and \eqref{eq:P_a_nhr_2} into \eqref{eq:x_a_nk} implies
\begin{align}
    \epsilon^{\ell}(x^a_{nk}) =&
    \frac{p^a_k(1-I^a_{-n,k})}{s^a_k} 
    \Biggl[ \phi^a \lambda^a_{n} y^a_{n} \frac{\partial D^{a,1}_{n}}{\partial P^a_{hr}} \notag \\
    &+\sum_{n'\in\mathcal{N}, n'\ne n} \phi^a \lambda^a_{n'} y^a_{n'} 
    \left( \frac{\partial D^{a,1}_{n'}}{\partial P^a_{hr}} + D^t_{n'} \right) \Biggr] \notag \\
    &- \frac{p^a_k I^a_{-n,k}}{s^a_k} \phi^a \lambda^a_n y^a_n D^t_n, \label{eq:storage_efficiency_2}
\end{align}
where
\begin{align}
    \frac{\partial D^{a,1}_n}{\partial P^a_{hr}} = &- \frac{w^a}{f^a_n} 
    - \frac{(\lambda^a_n R^a)^2 w^a ((1-P^a_{hr})w^a + w^s)^2}{2 f^a_n (f^a_n - \lambda^a_n R^a ((1-P^a_{hr}) w^a + w^s))^2} \notag \\
    &- \frac{f^a_n (\lambda^a_n R^a)^2 (1-P^a_{hr})(1+P^a_{hr}) w^a}{2 (\mu^{a,0}_n)^2 (f^a_n - \lambda^a_n R^a ((1-P^a_{hr}) w^a + w^s))^2} \notag \\
    &- \frac{f^a_n \lambda^a_n R^a P^a_{hr}}{(\mu^{a,0}_n)^2 (f^a_n - \lambda^a_n R^a ((1-P^a_{hr}) w^a + w^s))} \notag \\
    &- \frac{\lambda^a_n R^a w^a ((1-P^a_{hr})w^a + w^s)}{f^a_n (f^a_n - \lambda^a_n R^a ((1-P^a_{hr}) w^a + w^s))}.
    \label{eq:cache_delay_derivative}
\end{align}
Based on $\epsilon^{\ell}(x^a_{nk})$, the sample-average marginal storage efficiency 
can be expressed as
\begin{equation}
    \overline\epsilon_{nk}^a(\bm{x}_n)
      =\frac{1}{L}\sum_{\ell\in\mathcal L_r}\epsilon^\ell(x_{nk}^a).
    \label{eq:average_storage_efficiency}
\end{equation}
To account for the update cost, we further define the marginal storage efficiency of the objective as
\begin{equation}
    \widetilde\epsilon_{nk}^a(\bm{x}_n)
      =\frac{\partial J_n}{\partial\sigma_{nk}^a}
      =\overline\epsilon_{nk}^a(\bm{x}_n)+\delta_{nk}^a.
    \label{eq:update_aware_efficiency}
\end{equation}
Equation~\eqref{eq:update_aware_efficiency} measures the marginal change in the objective per unit of allocated storage, accounting for both response time and cache update cost.
% A smaller value therefore indicates a greater marginal benefit from allocating additional storage.
Moreover, one can verify that $D_n^{a,1}$ is convex in $P_{hr}^a$ whenever the queueing stability condition in \eqref{cons:stability} is satisfied with $y_n^a=1$.
Since the hit rates are affine in $\bm{x}_n$ and the transfer-delay and update-cost terms are linear, the relaxed caching subproblem is convex.
Therefore, its optimal solution can be characterized by the Karush--Kuhn--Tucker (KKT) conditions, yielding the following theorem.

\begin{theorem}
    Suppose that the system remains stable at $\bm{x}_n=\bm{0}$ under the fixed caching decisions at other BSs and the fixed sample-specific scheduling and cache-searching decisions. 
    For $C_n^s>0$, a feasible solution $\bm{x}_n^*$ is optimal for $P_n^c$ if and only if there exists $B^*\leq0$ such that
    \begin{equation}
    \begin{cases}
      \widetilde\epsilon_{nk}^a(\bm{x}_n^*)=B^*, &0<x_{nk}^{a,*}<1,\\
      \widetilde\epsilon_{nk}^a(\bm{x}_n^*)\leq B^*, &x_{nk}^{a,*}=1,\\
      \widetilde\epsilon_{nk}^a(\bm{x}_n^*)\geq B^*, &x_{nk}^{a,*}=0,
    \end{cases}
    \label{eq:cache_kkt_threshold}
    \end{equation}
    for all $a,k$, together with
    \begin{equation}
      B^*\left(C_n^s-\sum_{a,k}s_k^a x_{nk}^{a,*}\right)=0.
      \label{eq:cache_capacity_complementarity}
    \end{equation}
    \label{theorem:optimality}
\end{theorem}
\begin{proof}
   Please see Appendix~\ref{appendix:2}.
\end{proof}
Based on Theorem~\ref{theorem:optimality}, we solve $P_n^c$ by searching for the marginal storage efficiency threshold $B^*$ and determining the corresponding caching decisions.
Specifically, for $B\leq0$, we define
\begin{equation}
    g(B)=\mathop{\rm arg\,min}_{\bm{0}\leq\bm{x}_n\leq\bm{1}}
       \left\{J_n(\bm{x}_n)-B S(\bm{x}_n)\right\}
    \label{eq:threshold_mapping}
\end{equation}
where $S(\bm{x}_n)=\sum_{a,k}s_k^a x_{nk}^a$ is the total storage used at BS $n$.
By \eqref{eq:update_aware_efficiency}, the derivative of the objective in \eqref{eq:threshold_mapping} with respect to $x_{nk}^a$ is
\begin{equation*}
    \frac{\partial [J_n(\bm{x}_n)-B S(\bm{x}_n)]}{\partial x_{nk}^a}
    =s_k^a\bigl(\widetilde\epsilon_{nk}^a(\bm{x}_n)-B\bigr).
\end{equation*}
At a minimizer over $0\leq x_{nk}^a\leq1$, this derivative must be zero for $0<x_{nk}^a<1$, nonpositive for $x_{nk}^a=1$, and nonnegative for $x_{nk}^a=0$.
Thus, solving \eqref{eq:threshold_mapping} yields caching decisions
that satisfy \eqref{eq:cache_kkt_threshold} for the given threshold $B$.

Next, we examine how the total storage usage varies with $B$.
For any $B_1<B_2$, let $\bm{x}_1$ and $\bm{x}_2$ be the corresponding optimal solutions of \eqref{eq:threshold_mapping}.
By their optimality, we have
\begin{align*}
    J_n(\bm{x}_1)-B_1S(\bm{x}_1)
    &\leq J_n(\bm{x}_2)-B_1S(\bm{x}_2),\\
    J_n(\bm{x}_2)-B_2S(\bm{x}_2)
    &\leq J_n(\bm{x}_1)-B_2S(\bm{x}_1).
\end{align*}
Adding these inequalities yields
\begin{equation*}
    (B_2-B_1)\bigl(S(\bm{x}_2)-S(\bm{x}_1)\bigr)\geq0.
\end{equation*}
Since $B_2>B_1$, we have $S(\bm{x}_2)\geq S(\bm{x}_1)$.
Thus, the total storage usage is nondecreasing with respect to $B$, 
which enables bisection search of $B^*$.

The search starts at $B=0$.
If the resulting caching decisions satisfy the storage constraint, they are optimal for $P_n^c$.
Otherwise, we search for a negative threshold to reduce storage usage.
The search interval is initialized with $B_r=0$ and a lower bound $B_l$ at which caching no results is optimal.
At each iteration, we evaluate the midpoint: if the resulting storage usage exceeds $C_n^s$, we decrease the threshold; otherwise, we increase it, as specified in Algorithm~\ref{alg:caching}.
The iterations continue until the interval width satisfies $B_r-B_l\leq\xi$, where $\xi>0$ is the prescribed search tolerance.

We next develop an efficient algorithm to compute a minimizer in $g(B)$ for a given threshold $B$. 
Intuitively, in \eqref{eq:threshold_mapping}, $-B$ serves as an additional cost per unit of storage. 
Thus, a more negative $B$ makes storage more costly and favors solutions with lower storage usage. 
For a fixed $B$, the caching decisions can be optimized separately for each application. 
Results already cached at other BSs can be handled directly by comparing their marginal storage efficiencies with $B$. 
For the remaining results, we establish their caching priorities based on their popularity, sizes, and update costs under the given threshold $B$, and then obtain the caching decisions via a one-dimensional bisection search.

% The next step is to efficiently compute a minimizer in $g(B)$ for a given threshold $B$.
% In \eqref{eq:threshold_mapping}, $-B$ serves as an additional cost per unit of storage.
% A more negative $B$ makes storage usage more costly and therefore favors lower storage usage.
% For a fixed $B$, the caching decisions can be optimized separately for each application.
% Results already cached at other BSs can be handled by directly comparing their marginal storage efficiencies with $B$.
% For the remaining results, we determine their caching priorities
% based on their popularity, sizes, and update costs for the given
% threshold $B$, and then use a one-dimensional bisection search
% to determine the caching decisions.

Let $\mathcal K_1^a=\{k:I_{-n,k}^a=1\}$ and $\mathcal K_0^a=\{k:I_{-n,k}^a=0\}$ denote the results already cached and not cached at other BSs, respectively.
We first consider $k\in\mathcal K_1^a$.
Caching such a result at BS $n$ does not change the overall hit rate, but avoids the delay of retrieving it from another BS.
Let
\begin{equation*}
    v_n^a=\frac{1}{L}\sum_{\ell\in\mathcal L_r}
           \phi^a\lambda_n^{a,\ell}y_n^{a,\ell}D_n^t
\end{equation*}
denote the sample-average weighted remote-transfer delay per unit of remote hit probability at BS $n$.
By \eqref{eq:storage_efficiency_2} and \eqref{eq:update_aware_efficiency}, the marginal storage efficiency is the constant
\begin{equation*}
    \widetilde\epsilon_{nk}^a
      =\delta_{nk}^a-\frac{p_k^a v_n^a}{s_k^a},
      \qquad k\in\mathcal K_1^a.
\end{equation*}
Thus, the objective in \eqref{eq:threshold_mapping} is linear in $x_{nk}^a$ for these results.
Consequently, the optimal solution is $x_{nk}^a=1$ if $\widetilde\epsilon_{nk}^a<B$ and $x_{nk}^a=0$ otherwise.

We next consider $k\in\mathcal K_0^a$.
Caching these results increases the overall hit rate, so their marginal storage efficiencies depend on their joint caching decisions.
Define
\begin{equation*}
    P_0^a=\sum_k p_k^a I_{-n,k}^a,
    \qquad h_a(\bm{x}_n)=\sum_{k\in\mathcal K_0^a}p_k^a x_{nk}^a,
\end{equation*}
where $P_0^a$ is the hit rate provided by other BSs and $h_a$ is the additional hit rate provided by BS $n$.
The overall hit rate is therefore $P_{hr}^a=P_0^a+h_a(\bm{x}_n)$.
Let $Q_a(h)$ collect the computation-delay and remote-transfer terms that depend on this additional hit rate.
Its derivative is the sample average of the bracketed term in \eqref{eq:storage_efficiency_2}, evaluated at $P_{hr}^a=P_0^a+h$, and $Q_a$ is convex.
For a given $h$, these delay terms are fixed, so we only need to minimize the sum of the linear terms $(\delta_{nk}^a-B)s_k^a x_{nk}^a$ in \eqref{eq:threshold_mapping}.
To determine which results to cache first, we define the following sorting key for $p_k^a>0$:
\begin{equation}
    \theta_k^a(B)=\frac{(\delta_{nk}^a-B)s_k^a}{p_k^a}.
    \label{eq:cache_sorting_key}
\end{equation}
We sort these results in ascending order of $\theta_k^a(B)$,
denoting their indices by $\pi_1,\ldots,\pi_m$, where $m$ is
the number of results with positive popularity in $\mathcal K_0^a$.
Let $H_i=\sum_{j=1}^i p_{\pi_j}^a$ be the cumulative popularity of the first $i$ results, with $H_0=0$.
For a prescribed value $h$ of $h_a(\bm{x}_n)$, the sum of these linear terms is minimized by caching results in this order.
Specifically, if $h\in[H_{i-1},H_i]$, the results preceding $\pi_i$ are fully cached, $\pi_i$ may be partially cached, and the remaining results are not cached.
The corresponding minimum linear cost is
\begin{equation}
    C_a(h;B)=\sum_{j<i}(\delta_{n\pi_j}^a-B)s_{\pi_j}^a
          +\theta_{\pi_i}^a(B)(h-H_{i-1}).
    \label{eq:cache_hit_cost}
\end{equation}
Therefore, determining the caching decisions for $\mathcal K_0^a$ reduces to the following one-dimensional convex problem:
\begin{equation}
    \min_{0\leq h\leq H_m}\;Q_a(h)+C_a(h;B).
    \label{eq:cache_scalar_problem}
\end{equation}
Within each interval $H_{i-1}<h<H_i$, the derivative of the objective in \eqref{eq:cache_scalar_problem} is $Q_a'(h)+\theta_{\pi_i}^a(B)$.
Since this convex objective has a nondecreasing subgradient, the optimal $h$ can be found by bisection.
Once $h$ is obtained, the caching decisions are given by $x_{n\pi_i}^a=[(h-H_{i-1})/p_{\pi_i}^a]_0^1$, where $[z]_0^1=\min\{1,\max\{0,z\}\}$.
% Notice that the ordering in \eqref{eq:cache_sorting_key} depends on both the previous caching decisions and the threshold $B$, so it must be updated during the threshold search.
Finally, we obtain binary caching decisions by rounding all fractional entries of the feasible relaxed solution down to zero.

\begin{algorithm}[t]
    \caption{Caching Optimization}
    \label{alg:caching}
    \begin{algorithmic}[1]
        \REQUIRE $\bm{x}$, $\bm{x}_n^{\mathrm{old}}$, tolerances $\xi,\zeta>0$;
        \IF{the system is unstable when $\bm{x}_n=\bm{0}$}
            \RETURN the current caching decisions at BS $n$;
        \ENDIF
        \STATE Obtain $\bm{x}_r = g(0)$ using Algorithm~\ref{alg:g};
        \IF{$S(\bm{x}_r)\leq C_n^s$}
            \RETURN $\hat{\bm{x}}_n=\lfloor\bm{x}_r\rfloor$;
        \ELSE
            \STATE $B_l=\min\{0,\min_{a,k}\widetilde\epsilon_{nk}^a(\bm{0})\}$, $B_r=0$, $\bm{x}_l=\bm{0}$;
            \WHILE{$B_r-B_l>\xi$}
                \STATE $B_m=(B_l+B_r)/2$;
                \STATE Obtain $\bm{x}_m = g(B_m)$ using Algorithm~\ref{alg:g};
                \IF{$S(\bm{x}_m)\leq C_n^s$}
                    \STATE $(B_l,\bm{x}_l)=(B_m,\bm{x}_m)$;
                \ELSE
                    \STATE $B_r=B_m$;
                \ENDIF
            \ENDWHILE
            \RETURN $\hat{\bm{x}}_n=\lfloor\bm{x}_l\rfloor$;
        \ENDIF
    \end{algorithmic}
\end{algorithm}

\begin{algorithm}[t]
    \caption{Threshold Mapping}
    \label{alg:g}
    \begin{algorithmic}[1]
        \REQUIRE $B\leq0$, tolerance $\zeta$;
        \FOR{each application $a$}
            % \STATE Set all zero-popularity entries to zero;
            \FOR{each $k\in\mathcal K_1^a$ with $p_k^a>0$}
                \STATE Set $x_{nk}^a=1$ if $\delta_{nk}^a-p_k^a v_n^a/s_k^a<B$, and zero otherwise;
            \ENDFOR
            \STATE Sort positive-popularity $k\in\mathcal K_0^a$ by \eqref{eq:cache_sorting_key}; \label{alg:sort_line}
            \STATE Form the prefix sums $H_i$ and the piecewise-linear function $C_a(h;B)$;
            \STATE Solve \eqref{eq:cache_scalar_problem} by subgradient bisection until the search interval has width at most $\zeta$;
            \STATE Set $x_{n\pi_i}^a=[(h-H_{i-1})/p_{\pi_i}^a]_0^1$ for all $i=1,\ldots,m$;
        \ENDFOR
        \RETURN $\bm{x}_n$;
    \end{algorithmic}
\end{algorithm}

The detailed algorithm is summarized in Algorithms~\ref{alg:caching} and~\ref{alg:g}, 
where $\xi$ and $\zeta$ specify tolerances of bisection.
Let $K=\sum_a|\mathcal K^a|$.
Since $H_m\leq1$, the search over $h\in[0,H_m]$ requires $O(\log(1/\zeta))$ iterations.
Hence, the time complexity of Algorithm~\ref{alg:g} is
\begin{equation*}
    C_g=O\bigl(K\log K+A\log(1/\zeta)(LN+\log K)\bigr),
\end{equation*}
where the first term refers to the time complexity of sorting in line \ref{alg:sort_line}.
Based on this result, the time complexity of Algorithm~\ref{alg:caching} can be easily derived
as $C_c=O\bigl(C_g\log(1/\xi)\bigr)$.

The following theorem shows that the performance of the rounded caching decision is close to the optimal one for the single-BS caching subproblem.
Let $s^{max}=\max_{a,k}s_k^a$ denote the maximum size of a computation result.
\begin{theorem}[Asymptotic optimality]
    Under the stability assumption in Theorem~\ref{theorem:optimality}, let $J_n^{\rm rel}>0$ and $J_n^{\rm opt}$ be the optimal values of the relaxed and binary caching subproblems, respectively.
    Let $\hat J_n=J_n(\hat{\bm{x}}_n)$ denote the objective value obtained by rounding down an exact relaxed optimum.
    Then
    \begin{equation}
        1\leq\frac{\hat J_n}{J_n^{\rm opt}}
          \leq1+\kappa_n\frac{(A+1)s^{max}}{C_n^s},
        \label{eq:rounding_relative_bound}
    \end{equation}
    where $\kappa_n$ is the normalized marginal sensitivity defined in Appendix~\ref{appendix:3}.
    If $\kappa_n$ remains uniformly bounded, the obtained solution is asymptotically optimal as $s^{max}/C_n^s\to0$.
    \label{theorem:performance_gap}
\end{theorem}
\begin{proof}
   Please see Appendix~\ref{appendix:3}.
\end{proof}

Notice that the above algorithm and theoretical analysis only apply to the caching subproblem at a single BS. 
To obtain network-wide caching decisions, we iteratively optimize the caching decision at each BS while holding those at all other BSs fixed. 
Since the solution involves numerical approximation and rounding, the resulting caching decision is not guaranteed to decrease the original frame-level objective in \eqref{subproblem:frame}. 
We therefore introduce an acceptance criterion to only accept those rounded candidates
that decrease the frame-level objective.

%Based on such observation, we can prove that the optimal caching decisions satisfies the following property:
%\begin{lemma}
    %For any given \textcolor{red}{caching decisions $\bm{x}_{-n}$ on other BSs},
    %there exists an optimal solution $\bm{x}^{*}_{n} = (x^{a,*}_{nk})_{a\in\mathcal{A}, k\in\mathcal{K}^a}$ of \textcolor{red}{$P^c_n$} such that for each application $a$,
    %there is at most one typical input $k'$ whose caching decision is fractional, i.e. $x^{a,*}_{nk} \in \{0,1\}$ for all $k \ne k'$.
    %\label{lemma:optimal_scheduling}
%\end{lemma}
% lemma 1 shows that the resulted optimal solution can be easily transformed to a binary solution and the performance is close

\subsection{Scheduling Subproblem} \label{subsection:scheduling}
Compared to the caching subproblem, the scheduling subproblem \eqref{subproblem:scheduling} is much easier because the decision
variables are real numbers so gradient-based methods can be used.
The same scheduling algorithm is used to optimize the sample-specific scheduling decisions in problem \eqref{subproblem:frame}
and the actual scheduling decisions in each slot.
Therefore, we consider a generic workload state and omit the frame, slot, and sample indices in this subsection.
Since \eqref{subproblem:scheduling} is a constrained optimization problem, we can apply the Projected Gradient Descent (PGD) algorithm
to find a local minimum.
Compared to the normal gradient descent, PGD takes an extra operation that projects the updated solution to the feasible
region so that all constraints are satisfied.
In our problem, let $(\bm{\lambda}^i, \bm{f}^i)$ be the obtained solution after the $i$-th iteration step, then the next update step is
\begin{equation*}
    (\bm{\lambda}^{i+1}, \bm{f}^{i+1}) = Proj\left[ (\bm{\lambda}^i,\bm{f}^i) - \theta^i \nabla D(\bm{\lambda}^i, \bm{f}^i) \right],
\end{equation*}
where $Proj[\cdot]$ is the projection operator that maps the input to the closest point in the feasible region, $\theta^i$ is the learning rate, 
and $D(\bm{\lambda}, \bm{f}) = \sum_{a\in\mathcal{A}} \phi^a D^a$ is the objective of \eqref{subproblem:scheduling}.

To implement the projection operator, we first replace $f^a_n$ with $\hat{f}^a_n = f^a_n / C^c_n$.
Moreover, the system response time is decreasing with respect to $f^a_n$ so we can assume constraint \eqref{lmp:computing} is tight.
Therefore, constraint \eqref{lmp:computing} can be equivalently expressed as 
\begin{equation}
    \sum_{a\in\mathcal{A}} \hat{f}^a_n = 1, \quad \forall n.
    \label{cons:computing_eq}
\end{equation}
Combining \eqref{cons:scheduling}, \eqref{cons:computing_eq}, and \eqref{lmp:variable_ylf} we know that the feasible regions
of $(\lambda^a_n)_{n\in\mathcal{N}}$ and $(\hat{f}^a_n)_{a\in\mathcal{A}}$ are both probability simplex~\cite{krstovski2013efficient}.
According to \cite{wang2013projection}, projection to the probability simplex can be computed in $O(d\log d)$ time,
where $d$ is the dimension of the probability simplex.
Hence, in our problem, the time complexity of the projection operator is $O(NA\log(NA))$.

After implementing the projection operator, we still need to 1) decide the learning rate $\theta^i$ and 
2) guarantee the obtained solution satisfies constraint \eqref{cons:stability}.\footnote{Recall that we have relaxed this constraint previously.}
The two issues can be simultaneously addressed by the backtracking line search \cite{armijo1966minimization}.
Intuitively, the backtracking method starts with a relatively large step size and then iteratively shrinks it by a contraction
factor until a sufficient decrease of the objective function is observed.
Since our feasible region is convex, the line between $(\bm{\lambda}^i, \bm{f}^i)$ and $(\bm{\lambda}^{i+1}, \bm{f}^{i+1})$
also lies in the feasible region.
Therefore, we can first choose a relatively large $\theta^i$ and then backtracking along the direction
$\bm{p}^i = (\bm{\lambda}^{i+1}, \bm{f}^{i+1}) - (\bm{\lambda}^i, \bm{f}^i)$.
Let $\alpha\in(0,1)$ and $\beta\in(0,1)$ be the control parameters and $j=0$ be the iteration counter, 
then the standard backtracking algorithm repeatedly increments $j$ until
\begin{equation}
    D(\bm{\lambda}^i, \bm{f}^i) - D((\bm{\lambda}^i, \bm{f}^i) + \beta^j \bm{p}) \geq - \alpha \beta^j \nabla D(\bm{\lambda}^i, 
    \bm{f}^i)^T \bm{p}.
    \label{ineq:backtracking}
\end{equation}
Theoretical analysis guarantees that such $j$ always exists and is denoted by $j_1$.
In addition to the decrease of objective value, an extra requirement of our problem is that the obtained solution 
$(\bm{\lambda}^i, \bm{f}^i) + \beta^j \bm{p}$ should also satisfy constraint \eqref{cons:stability}.
Let $\mathcal{S}$ be the set of $(\bm{\lambda}, \bm{f})$ that satisfies \eqref{cons:stability}.
Notice that $\mathcal{S}$ is an open set, so there exists a neighborhood $U$ of 
$(\bm{\lambda}^i, \bm{f}^i)$ such that $U \subseteq \mathcal{S}$.
Moreover, since $\beta \in (0,1)$,
we can always find a sufficiently large $j_2$ such that for all $j \geq j_2$ we have $(\bm{\lambda}^i, \bm{f}^i) + \beta^j \bm{p} \in U$.
Let $j^* = \max(j_1, j_2)$, then $(\bm{\lambda}^i, \bm{f}^i) + \beta^{j^*} \bm{p}$ satisfies \eqref{ineq:backtracking}
and \eqref{cons:stability} simultaneously.
According to \eqref{eq:delay_no_cache}, \eqref{delay:cache} and \eqref{eq:D^a_n},
$D^a_n$ increases to infinity when $(\bm{\lambda}, \bm{f})$ approaches the boundary of $\mathcal{S}$.
Therefore, the evolution trace of $(\bm{\lambda}^i, \bm{f}^i)$ will stay away from the boundary of $\mathcal{S}$
and makes $j_2$ very small in practice.
As a result, the backtracking algorithm usually terminates in several iterations if we choose a proper $\alpha$ and $\beta$.

After each PGD iteration, we update $\bm{y}$ according to \eqref{eq:optimal_y} and retain the update only when the objective value decreases.
This process terminates when $\bm{y}$ no longer changes or the maximum number of iterations is reached.
For the actual scheduling problem in slot $t>1$, we use $(\bm{\lambda}[t-1],\bm{f}[t-1])$ as the initial point.
Such a warm-start strategy substantially reduces the required number of iterations when the task arrival rates vary smoothly across consecutive slots.

\subsection{Online Implementation and Analysis}

The complete two-timescale algorithm is summarized in Algorithm \ref{alg:two_timescale}.
At the beginning of each frame, we first estimate the task popularity and construct the workload sample set from recent observations.
The sample-specific scheduling decisions and caching decisions are then optimized alternately.
For each BS, Algorithms \ref{alg:caching} and \ref{alg:g} generate a binary candidate using the caching objective and the fixed previous-frame cache state.
The candidate is accepted only if it reduces the objective of problem \eqref{subproblem:frame}.
After the frame-level iteration terminates, the resulting caching decision is deployed and fixed within the frame.
In each slot, the current task arrival rates are observed and the scheduling problem is solved with a warm start.

\begin{algorithm}[t]
    \caption{Two-Timescale Online Algorithm}
    \label{alg:two_timescale}
    \begin{algorithmic}[1]
        \STATE Initialization: caching decision $\bm{x}[0]$ and feasible scheduling decisions;
        \FOR{$r\in\mathcal{R}$}
            \STATE Estimate $\bm{p}[r]$ and set $\mathcal{L}_r$ to the indices of the $L$ most recently observed slots;
            \STATE Fix $\bm{x}^{\mathrm{old}}=\bm{x}[r-1]$; initialize $\bm{x}^{0}[r]=\bm{x}^{\mathrm{old}}$ and $q=0$;
            \REPEAT
                \FOR{$\ell\in\mathcal{L}_r$}
                    \STATE Optimize $\bm{y}^{\ell}[r],\bm{\lambda}^{\ell}[r],\bm{f}^{\ell}[r]$ using PGD and backtracking;
                \ENDFOR
                \STATE $\bm{x}^{q+1}[r]=\bm{x}^{q}[r]$;
                \FOR{$n\in\mathcal{N}$}
                    \STATE Obtain $\hat{\bm{x}}_n$ using Algorithm~\ref{alg:caching} with $\bm{x}_n^{\mathrm{old}}$ fixed;
                    \STATE Evaluate the candidate after updating $\{\bm{y}^{\ell}[r]\}_{\ell\in\mathcal{L}_r}$ using \eqref{eq:optimal_y};
                    \IF{the candidate reduces the objective of problem \eqref{subproblem:frame}}
                        \STATE Replace the caching decisions at BS $n$ in $\bm{x}^{q+1}[r]$ with $\hat{\bm{x}}_n$;
                        \STATE Retain the updated sample-specific cache-searching decisions;
                    \ENDIF
                \ENDFOR
                \STATE $q=q+1$;
            \UNTIL{the objective converges or the maximum number of iterations is reached}
            \STATE Deploy $\bm{x}[r]=\bm{x}^{q}[r]$;
            \FOR{$t\in\mathcal{T}_r$}
                \STATE Observe $\bm{R}[t]$ and warm-start from the preceding slot;
                \STATE Optimize $\bm{y}[t],\bm{\lambda}[t],\bm{f}[t]$ using PGD and backtracking;
            \ENDFOR
        \ENDFOR
    \end{algorithmic}
\end{algorithm}

Next, we analyze the convergence of the proposed algorithm. At the frame level, the backtracking line search ensures that each scheduling update does not increase $\overline{D}[r]$, while the acceptance rule ensures that each caching update either decreases the frame-level objective or leaves the current decision unchanged. Since the objective is lower bounded, the resulting objective sequence converges, and the algorithm reaches a fixed point when no further scheduling or caching update yields an improvement. 
% This fixed point is not necessarily globally optimal due to the non-convexity of the problem and the binary reconstruction of caching decisions.
At the slot level, for each fixed cache-searching mode, the alternating updates of $\bm{y}[t]$, $\bm{\lambda}[t]$, and $\bm{f}[t]$, together with PGD and backtracking, generate a non-increasing objective sequence. Upon convergence, the resulting solution is stationary with respect to the continuous variables and locally optimal with respect to $\bm{y}[t]$.

Let $I_f$ and $I_s$ denote the maximum numbers of frame-level AM iterations and scheduling PGD iterations, respectively,
and let $C_c$ be the caching-update complexity derived in Section~\ref{subsection:caching}.
The complexity of the frame-level planning process is
$O(I_f(NC_c+LI_sNA\log(NA)))$,
while the complexity of each slot-level update is $O(I_sNA\log(NA))$.

\section{Performance Evaluation} \label{section:simulation}

In this section, we evaluate the effectiveness of the proposed algorithm in adapting computation reuse to time-varying workloads through numerical simulations.

\subsection{Simulation Setup}

The proposed algorithm, denoted by Prop, is evaluated over a
24-hour trace comprising 48 frames. Each frame lasts 30 minutes and contains
$T=60$ slots of 30 seconds. Unless otherwise specified, the network consists
of $N=10$ BSs serving $A=5$ applications. Table~\ref{tab:sim_param}
summarizes the default parameters. %, whose ranges follow \cite{al2022promise,al2022reservoir}. 
Here, $\mathcal{U}[a,b]$ denotes a
uniform distribution, and $\mathcal{N}(\mu,\sigma)$ denotes a normal
distribution with mean $\mu$ and standard deviation $\sigma$.

We model temporal load variations as correlated multiplicative fluctuations around the baseline per-BS arrival rates. Specifically, we generate a stationary first-order autoregressive process~\cite{brockwell2016introduction},
$z_n^a[t]=\phi z_n^a[t-1]+\epsilon_n^a[t]$, where $\phi=0.95$ and $\epsilon_n^a[t]$ denotes an independent zero-mean Gaussian noise. 
The stationary variance is set to $\sigma_z^2=\ln(1+0.2^2)$ by choosing the noise variance as $(1-\phi^2)\sigma_z^2$. We then define the positive load factor as
\begin{equation*}
    m_n^a[t]=\exp\!\left(z_n^a[t]-\frac{\sigma_z^2}{2}\right),
\end{equation*}
such that $\mathbb{E}[m_n^a[t]]=1$ and its coefficient of variation is $\sqrt{\exp(\sigma_z^2)-1}=0.2$. Each baseline arrival rate is scaled by $m_n^a[t]$ to obtain the corresponding time-varying load. To capture abrupt changes in demand intensity and spatial distribution, we introduce load bursts and a $30\%$ redistribution of BS-side traffic at frames 17 and 33.

Input popularity evolves independently of the arrival-rate process. Reusable inputs account for $0.8$ of the total probability mass and follow a Zipf distribution with an exponent drawn from $\mathcal{U}[0.8,1.2]$. For each application, the top $10\%$ of inputs constitute a persistent hot set. To capture both gradual popularity drift and abrupt changes in hot-set membership, we replace a fraction $\delta_p=0.075$ of the hot-set members in each ordinary frame and increase this fraction to $0.3$ at frames 17 and 33.

\begin{table}[tb]
\renewcommand{\arraystretch}{1.1}
\caption{Default Simulation Parameters}
\label{tab:sim_param}
\centering
% \scriptsize
\begin{tabularx}{0.8\columnwidth}{Xc}
\hline
\textbf{Parameter} & \textbf{Value}\\
\hline
Computational capacity $C^c_n$ & $\mathcal{U}[2,8]$ GHz \\
Storage capacity $C^s_n$ & $\mathcal{U}[1,4]$ GB \\
Number of typical inputs $K^a$ & $\mathcal{U}[10\mathrm{K},50\mathrm{K}]$ \\
Task workload $w^a$ & $\mathcal{U}[200,600]$ Mcycles \\
Cache-search workload $w^s$ & $25$ Mcycles \\
Arrival rate $R^a_n$ & $\mathcal{U}[0.5,1.5]$ tasks/s \\
% NoC arrival-rate scale & $0.5$ \\
Result size $s^a_k$ & $\mathcal{N}(0.2,0.06)$ MB \\
Transmission delay $D^t_n$ & $\mathcal{U}[10,30]$ ms \\
% Slot/frame/experiment duration & $30$ s/$30$ min/$24$ h \\
% History length $(T,L)$ & $(60,60)$ slots \\
% AR coefficient/arrival CV & $0.95/0.2$ \\
% Ordinary/change hot-set replacement & $0.075/0.3$ \\
% Burst/BS-hotspot migration & $1.75/0.3$ \\
% Warm-up frames & $4$ \\
Update weight $\eta$ & $0.1$ \\
% AM/PGD iteration caps & $3/25$ \\
% Relative/bisection tolerance & $10^{-3}/10^{-9}$ \\
Backtracking parameters $(\alpha,\beta)$ & $(0.3,0.5)$ \\
% VaRe-C $(\epsilon,T_{\max})$ & $(0.1,12\ \text{frames})$ \\
\hline
\end{tabularx}
\end{table}

Prop estimates the demand from observations over the preceding $L=60$ slots. 
The frame-level AM and slot-level PGD procedures are limited to $3$ and $25$ iterations, respectively, with a relative-change stopping tolerance of $10^{-3}$. The first four frames are treated as a warm-up period and excluded from all aggregate performance metrics.
We compare Prop with the following policies:
\begin{itemize}
    \item \emph{Static}: the static policy from our conference
    version~\cite{he2024exploiting}, which determines the caching decisions
    using the initial observations and retains them throughout the experiment.
    \item \emph{VaRe-C}~\cite{zhu2026optimized}: a value-aware policy that
    ranks results by estimated reuse value and size, applies scaled dynamic
    programming to high-value results, and fills the remaining storage
    greedily. 
    % It discounts older results using Age of Data and proactively
    % refreshes selected results with idle computing resources. 
    % In our
    % implementation, result $(a,k)$ corresponds to a data block of size
    % $s_k^a$, with reuse saving $w^a/C_n^c+D_n^t$. VaRe-C replaces only the
    % frame-level cache policy and uses the same slot-level controller as
    % Prop.
    \item \emph{Look-ahead}: an oracle benchmark with access to the true task 
      popularity and future slot-level arrival rates within the current frame.
    \item \emph{No Reuse (NoR)}: adapts the peer-offloading mechanism
    in~\cite{chen2018computation} to our model, using the same slot-level
    scheduling algorithm as Prop while disabling computation reuse.
    \item \emph{No Collaboration (NoC)}: adapts the local computation reuse
    mechanism of CoxNet~\cite{bellal2021coxnet} to our model, restricting
    task processing and result reuse to the local BS. Each BS caches results
    by their popularity-to-size ratios and allocates computing resources
    in proportion to the local workload.
\end{itemize}
The default workload exceeds the processing capacity of NoC. We therefore evaluate this policy under a reduced workload of $0.5R_n^a[t]$ throughout the experiments.
Each setting is evaluated over 30 random seeds. 
The lines show the mean results, 
while the shaded regions indicate the 95\% confidence intervals.

% The primary metrics are the time-average weighted response time $\bar D$
% (in seconds), normalized cache update cost $\bar U$, and
% objective $\bar J=\bar D+\eta\bar U$.
% We also report newly written cache data and the fraction of unstable
% application--BS queues. Each setting is evaluated using 30 independent
% random seeds; policies at a given setting and seed share the same system
% instance and underlying demand trace, subject to the NoC load scaling.
% Within each algorithm--setting group, a run is excluded from aggregation
% if its post-warm-up unstable-queue fraction is nonzero. Reported means and
% approximate 95\% confidence intervals are computed from the remaining
% stable runs, using $1.96$ times the across-seed standard error for the
% interval half-width. The retained sample size $n$ is reported because it
% may differ across policies and settings. These statistics characterize
% performance conditional on queue stability; differences between retained
% seed sets are considered when interpreting the results.

\subsection{Convergence Behavior and Running Time}

\begin{figure}[t]
\centering
\subfloat[Frame-level AM.\label{fig:convergence_frame}]{%
\includegraphics[width=0.49\columnwidth]{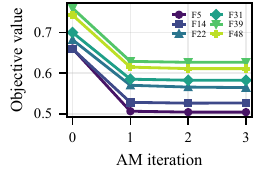}}\hfill
\subfloat[Slot-level PGD.\label{fig:convergence_slot}]{%
\includegraphics[width=0.49\columnwidth]{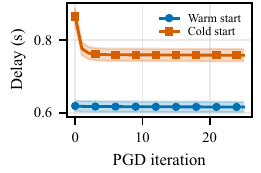}}
\caption{Convergence of the proposed algorithm.}
\label{fig:validation_combined}
\end{figure}

Fig.~\ref{fig:validation_combined} illustrates the convergence behavior for a representative simulation run. 
Specifically, Fig.~\ref{fig:convergence_frame} shows the surrogate objective for frames $\{5,14,22,31,39,48\}$, uniformly sampled from the evaluated portion of the trace. The objective decreases by 16.8\%--23.7\%, with most of the reduction achieved by the first AM iteration.

Fig.~\ref{fig:convergence_slot} compares warm and cold starts for slot-level PGD. The mean weighted response time decreases from 0.6190 to 0.6162~s with warm starts, compared with 0.8647 to 0.7584~s with cold starts. Under the prescribed iteration budget, warm starting achieves an 18.8\% lower final response time, demonstrating the benefit of reusing the preceding solution to provide a better initialization and improve finite-iteration performance.

Across the 30 runs, Prop requires an average of 11.44~s per 30-minute frame, including all 60 slot-level updates, whereas VaRe-C requires 1.23~s per frame with the same slot-level controller. Although Prop incurs higher optimization overhead, its average runtime remains negligible relative to the frame duration. Moreover, the current implementation is in Python, which can be substantially slower than lower-level languages such as C/C++ for computationally intensive workloads~\cite{benchmark}. Further runtime reductions can therefore be expected through optimized implementations in practical deployments.

\subsection{Performance under Dynamic Workloads}

\begin{table}[tb]
\renewcommand{\arraystretch}{1.1}
\caption{Performance and Ablation Results}
\label{tab:overall_performance}
\centering
\scriptsize
\begin{tabularx}{0.7\columnwidth}{Xrrr}
\hline
Algorithm & $\bar D$ (s) & $\bar U$ & $\bar J$ (s) \\
\hline
Prop & 0.9971 & 1.1649 & 1.1136 \\
VaRe-C & 1.1714 & 5.5558 & 1.7270 \\
Static & 1.9450 & 0 & 1.9450 \\
Look-ahead & 0.8109 & 0.1010 & 0.8210 \\
NoR & 4.4519 & 0 & 4.4519 \\
NoC & 2.2428 & 6.2571 & 2.8685 \\
\hline
$L=1$ & 1.2970 & 0.1459 & 1.3116 \\
Neglect update cost & 1.0278 & 5.7013 & 1.5979 \\
Cold-start & 1.5352 & 1.1649 & 1.6517 \\
\hline
\end{tabularx}
\end{table}

\begin{figure}[t]
\centering
\includegraphics[width=0.9\columnwidth]{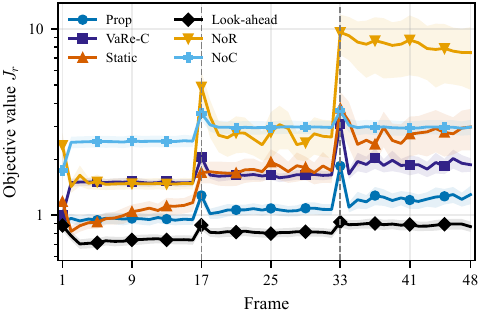}
\caption{Per-frame objective under dynamic demand.}
\label{fig:dynamic_response}
\end{figure}

We next evaluate the benefit of adapting computation reuse to time-varying demand. Table~\ref{tab:overall_performance} summarizes the aggregate performance, while Fig.~\ref{fig:dynamic_response} shows its temporal evolution on a logarithmic scale. Prop achieves $\bar D=0.9971$~s and $\bar J=1.1136$, reducing the mean objective by 35.5\%, 42.7\%, and 75.0\% compared with VaRe-C, Static, and NoR, respectively. These results demonstrate the benefit of jointly exploiting computation reuse, inter-BS collaboration, and dynamic cache adaptation under time-varying demand.

Even with its arrival rate halved, NoC yields $\bar D=2.2428$~s and $\bar J=2.8685$, highlighting the importance of inter-BS collaboration. In contrast, Look-ahead achieves the lowest objective of $0.8210$. The remaining gap between Prop and Look-ahead suggests that more accurate demand information could further improve performance.

The comparison with VaRe-C further highlights the importance of accounting for cache update costs. Although VaRe-C achieves a response time of $1.1714$~s, it incurs $14.09$~GB of cache writes per evaluated frame, compared with only $2.94$~GB for Prop. Consequently, its mean objective increases to $\bar J=1.7270$ under the common objective. By explicitly balancing the estimated response-time benefit of each cache update against its write cost, Prop achieves a lower overall objective while substantially reducing cache writes.

\subsection{Load and Popularity Sensitivity}
\begin{figure}[t]
\centering
\subfloat[Load sensitivity.\label{fig:load_sensitivity}]{%
\includegraphics[width=0.9\columnwidth]{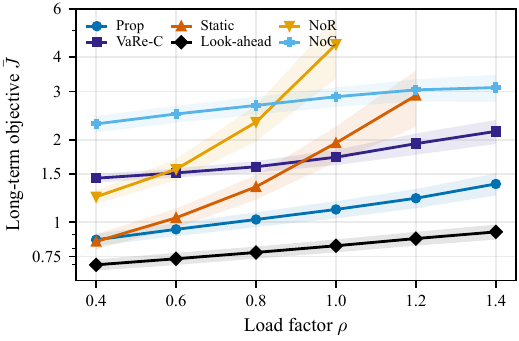}}\\
\subfloat[Popularity sensitivity. \label{fig:drift_sensitivity}]{%
\includegraphics[width=0.9\columnwidth]{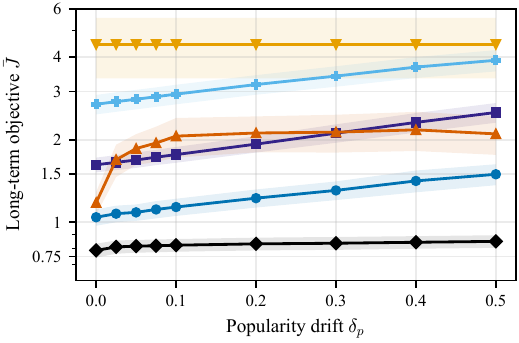}}
\caption{Sensitivity to workload variations.}
\label{fig:load_drift}
\end{figure}

This subsection evaluates the algorithms under varying load and popularity conditions. As shown in Fig.~\ref{fig:load_sensitivity}, performance generally degrades as the load increases, while Prop achieves the lowest objective among the evaluated causal algorithms. NoR and Static are particularly sensitive to increasing load. Their task arrival rates exceed the processing capacity in most runs at $\rho\geq1.2$ and $\rho=1.4$, respectively, and the corresponding results are therefore omitted. The widening gap between NoR and Prop demonstrates the increasing benefit of computation reuse under heavier workloads, whereas the degradation of Static shows that this benefit can be substantially diminished without timely cache adaptation to changing demand.

Fig.~\ref{fig:drift_sensitivity} varies the hot-set replacement fraction $\delta_p$ while retaining the 30\% migrations at frames 17 and 33 and the original load dynamics. As $\delta_p$ increases, the objectives of Prop, VaRe-C, and Static generally increase. Static exhibits a distinct saturation behavior: at low drift rates, initially popular results remain popular for longer, allowing the fixed cache to remain effective. As the drift increases, this persistence diminishes and the benefit of the initial cache decreases. Once the current hot set becomes weakly correlated with the initial one, however, further increasing the turnover rate no longer systematically reduces the overlap between the fixed cache and current demand, causing Static's performance to level off.
Look-ahead is only weakly affected by popularity drift because it adapts its caching decisions using the true popularity in each frame, whereas NoR is insensitive to popularity changes because it does not exploit cached results.

\subsection{Impact of Network Size}
\begin{figure}[t]
\centering
\includegraphics[width=0.9\columnwidth]{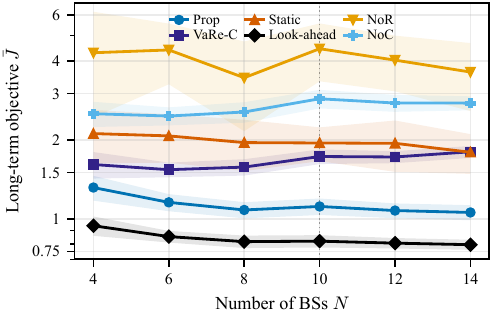}
\caption{Mean objective under different numbers of BSs.}
\label{fig:network_scale}
\end{figure}

Fig.~\ref{fig:network_scale} evaluates the impact of network size. A larger network provides more aggregate storage and enables BSs to access a wider range of cached results through collaboration. Consequently, Look-ahead, Prop, and Static generally achieve better performance as the network grows.
VaRe-C does not exhibit the same improvement because its caching decisions do not account for update costs. Although its mean response time decreases from 1.3568~s to 1.0804~s, the mean update cost increases from 2.5288 to 7.1976 as more BSs update their caches, raising the mean objective from 1.6097 to 1.8002. NoC, which restricts computation reuse and task processing to the local BS, cannot exploit the additional caching or computing resources provided by a larger network. Its response time is therefore relatively insensitive to network size, while its network-wide update cost increases with the number of BSs.
NoR exhibits greater variability because, without computation reuse, some BSs operate close to their processing capacity. In this regime, small capacity--load imbalances can cause sharp increases in queueing delay, allowing a few high-delay runs to substantially increase the mean response time and widen the confidence intervals.

\subsection{Two-Timescale Design and Ablation Study}

The ablations in Table~\ref{tab:overall_performance} assess the effects of demand history, cache update costs, and slot-level initialization. Using $L=60$ instead of $L=1$ reduces $\bar J$ from 1.3116 to 1.1136, achieving a 15.1\% improvement, 
which demonstrates the benefit of incorporating longer demand history. 
The \emph{Neglect update cost} variant sets $\eta=0$ during caching optimization while being evaluated under the same objective as Prop. This increases $\bar U$ from 1.1649 to 5.7013 and $\bar J$ to 1.5979, highlighting the importance of explicitly accounting for cache update costs. Finally, cold-start slot optimization increases $\bar D$ to 1.5352~s and $\bar J$ to 1.6517 without affecting the update cost, which shows the benefit of warm starts under a finite iteration budget.

\begin{figure}[t]
\centering
\includegraphics[width=0.9\columnwidth]{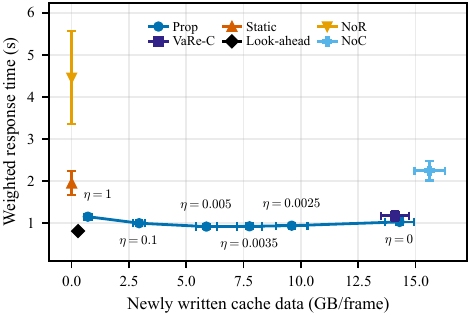}
\caption{Tradeoff between response time and cache writes.}
\label{fig:pareto}
\end{figure}

Fig.~\ref{fig:pareto} illustrates the tradeoff between response time and cache writes as the update penalty $\eta$ varies. We select $\eta\in\{0,0.0025,0.0035,0.005,0.1,1\}$ to provide approximately uniform coverage along the cache-write axis. Increasing $\eta$ consistently reduces cache writes, whereas the response time varies nonmonotonically, reflecting the discrete caching decisions and finite iteration budget of the alternating optimization. At $\eta=1$, the substantial reduction in cache writes is accompanied by increased response time, indicating insufficient adaptation to demand variations. The default setting $\eta=0.1$ achieves an objective 42.7\% lower than that of Static, demonstrating the benefit of balancing dynamic cache adaptation against its update cost.

\begin{figure}[t]
\centering
\includegraphics[width=0.9\columnwidth]{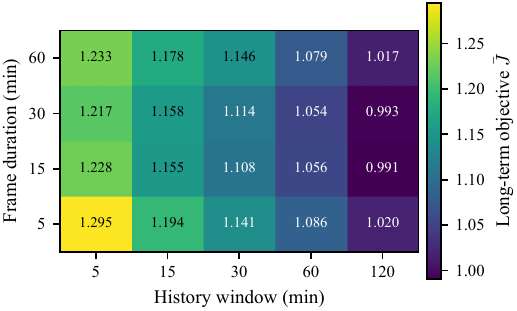}
\caption{Joint effect of frame and history-window durations.}
\label{fig:timescale_heatmap}
\end{figure}

% The ordinary-frame migration rate is adjusted to preserve 7.5\% hot-set
% turnover per 30 minutes, ensuring comparability in physical time.
Fig.~\ref{fig:timescale_heatmap} evaluates 20 combinations of frame and history-window durations. All combinations share the same 24-hour workload, a 120-minute warm-up period, and demand changes at hours 8 and 16. Within the tested range, the history-window duration has a greater impact on performance than the frame duration. 
The best combination, using a 15-minute frame and a 120-minute history, reduces the objective by 11.1\% relative to the default setting of a 30-minute frame and a 30-minute history.
Under the persistent hot-set workload, the improvement with longer histories suggests that the benefit of reduced estimation variability outweighs the effect of increased information staleness. 
In contrast, performance varies modestly with different frame durations, suggesting that the algorithm is relatively insensitive to the cache-update interval under the tested demand dynamics.

\section{Conclusions} \label{section:conclusion}
In this paper, we have proposed an online two-timescale computation reuse algorithm that coordinates caching and scheduling to balance response time and cache update cost in CEC networks.
Using recent workload observations, the algorithm uses bisection guided by marginal storage efficiency for caching and projected gradient descent for scheduling.
Our analysis establishes near-optimality for single-BS caching when result sizes are small relative to cache capacity.
Simulation results demonstrate that our algorithm adapts to dynamic demand 
and outperforms benchmark methods under various settings,
highlighting the benefits of coordinating computation reuse across timescales while accounting for update costs.

\appendices

\section{Proof of Theorem \ref{theorem:np}} \label{appendix:1}
    In the considered special case, we have only one BS and one application, so there is no need to offload tasks among BSs and share computational power among applications.
    As a result, the decision variables $\lambda^a_n$ and $f^a_n$ are redundant.
    In a stable system, the total task arrival rate should be lower than the total processing rate.
    This condition reduces to constraint \eqref{cons:stability} when $N=1$ so this constraint is also redundant.
    We further assume that the benefit of cache-searching is greater than the induced cost so we have $y^a_n=1$ in the optimal solution.
    In this case, it can be easily verified that the response time is optimized when we achieve the maximum hit rate.
    Therefore, the problem \eqref{lmp} is equivalent to the following problem:
    \begin{align*}
        \max\quad &P^{a}_{hr} = P^{a,n}_{lhr} = \sum_{k\in\mathcal{K}^a} x^a_{nk}p^a_k \\
        s.t.\quad &\sum_{k\in\mathcal{K}^a} x^a_{nk}s^a_k \leq C^s_n \\
        &x^a_{nk} \in \{0,1\}, \quad \forall k.
    \end{align*}
    The above problem is a standard 0-1 knapsack problem and is known to be NP-hard.
    
\section{Proof of Theorem \ref{theorem:optimality}} \label{appendix:2}
The stability assumption at $\bm{x}_n=\bm{0}$ ensures that the affine extension of the delay is finite and differentiable throughout the unit box.
The delay component is convex, as established in Section~\ref{subsection:caching}, and adding the linear update cost preserves convexity.
For $C_n^s>0$, the relaxed feasible set admits a strictly interior point, so the KKT conditions are necessary and sufficient.
Let $\nu\geq0$ be the capacity multiplier and $\alpha_{ak},\beta_{ak}\geq0$ the lower- and upper-bound multipliers.
Stationarity and complementarity give
\begin{align*}
    s_k^a\widetilde\epsilon_{nk}^a(\bm{x}_n^*)+\nu s_k^a
        -\alpha_{ak}+\beta_{ak}&=0,\\
    \alpha_{ak}x_{nk}^{a,*}=0,\quad
    \beta_{ak}(1-x_{nk}^{a,*})&=0,\\
    \nu\left(\sum_{a,k}s_k^a x_{nk}^{a,*}-C_n^s\right)&=0.
\end{align*}
Setting $B^*=-\nu$ yields \eqref{eq:cache_kkt_threshold} and \eqref{eq:cache_capacity_complementarity}.
% Conversely, these conditions and primal feasibility imply the KKT conditions, establishing optimality without requiring capacity saturation.

\section{Proof of Theorem \ref{theorem:performance_gap}} \label{appendix:3}
We establish the approximation bound by quantifying the objective increase caused by rounding down a relaxed optimum.
We first show that an exact relaxed optimum can be selected with at most $A+1$ fractional entries.
We then use convexity to bound the rounding loss in terms of the maximum result size and marginal sensitivity, yielding the relative bound in \eqref{eq:rounding_relative_bound}.

To select such an optimum, fix the additional hit rate $h_a$ of each application at its value in any relaxed optimum, which preserves all nonlinear terms in $J_n$.
The remaining optimization is linear over the unit box, subject to these $A$ hit-rate equalities and the capacity constraint.
An optimal extreme point therefore has at most $A+1$ fractional entries and remains optimal for the original relaxed subproblem.
Denote this solution by $\bm{x}_n^\circ$ and its fractional support by $\mathcal F_n$.

Let $\hat{\bm{x}}_n=\lfloor\bm{x}_n^\circ\rfloor$.
Rounding down preserves capacity feasibility and, under the stability assumption at $\bm{x}_n=\bm{0}$, stability, so $J_n^{\rm rel}\leq J_n^{\rm opt}\leq\hat J_n$.
Define
\begin{equation}
    M_n=\max_{(a,k)\in\mathcal F_n}
        [-\widetilde\epsilon_{nk}^a(\hat{\bm{x}}_n)]^+,
    \qquad \kappa_n=\frac{M_nC_n^s}{J_n^{\rm rel}},
    \label{eq:rounding_sensitivity}
\end{equation}
where $[u]^+=\max\{u,0\}$ and $M_n=0$ if $\mathcal F_n$ is empty.
Convexity and $\partial J_n/\partial x_{nk}^a=s_k^a\widetilde\epsilon_{nk}^a$ give
\begin{align}
    \hat J_n-J_n^{\rm rel}
    &\leq-\nabla J_n(\hat{\bm{x}}_n)^T
            (\bm{x}_n^\circ-\hat{\bm{x}}_n)\notag\\
    &\leq M_n\sum_{(a,k)\in\mathcal F_n}s_k^a x_{nk}^{a,\circ}
     \leq M_n(A+1)s^{max}.
    \label{eq:rounding_additive_bound}
\end{align}
Using $J_n^{\rm opt}\geq J_n^{\rm rel}>0$ yields \eqref{eq:rounding_relative_bound}, and the asymptotic result follows for fixed $A$ and uniformly bounded $\kappa_n$.

% Algorithm~\ref{alg:caching} has at most one fractional entry per application before rounding and requires no extreme-point selection.
% Applying the same argument to its feasible relaxed solution replaces $A+1$ with $A$ and adds $\Delta_n/J_n^{\rm rel}$ to the relative bound, with $M_n$ evaluated for that solution.

\bibliographystyle{IEEEtran}
\bibliography{ref}

\end{document}